\documentclass[letterpaper, 10 pt, conference]{ieeeconf}  

\IEEEoverridecommandlockouts                              
\usepackage[active]{srcltx} 
\usepackage{amsmath,amsxtra,amssymb}
\usepackage{latexsym,verbatim,fancybox}
\usepackage{cite,url}
\usepackage{algorithm2e}
\usepackage{psfrag}
\usepackage{graphicx}        

\newtheorem{theorem}{Theorem}

\newtheorem {proposition} {Proposition}

\title{\LARGE \bf
Security Games with Decision and Observation Errors
}

\author{Kien C. Nguyen, Tansu Alpcan, and Tamer Ba\c{s}ar
\thanks{This work was supported by Deutsche Telekom Laboratories and the Boeing Company.}
\thanks{Tamer Ba\c{s}ar and Kien C. Nguyen are with the Department of Electrical and Computer Engineering and the Coordinated Science Laboratory, University of Illinois at Urbana-Champaign, USA
        {\tt\small basar1@illinois.edu, knguyen4@illinois.edu}}%
\thanks{Tansu Alpcan is with Deutsche Telekom Laboratories, Technical University of Berlin, Berlin, Germany
        {\tt\small tansu.alpcan@telekom.de}}%
}

\begin{document}

\maketitle
\thispagestyle{empty}
\pagestyle{empty}

\begin{abstract}

We study two-player security games which can be viewed as sequences of nonzero-sum matrix games played by an Attacker and a Defender. The evolution of the game is based on a stochastic fictitious play process. Players do not have access to each other's payoff matrix. Each has to observe the other's actions up to present and plays the action generated based on the best response to these observations. However, when the game is played over a communication network, there are several practical issues that need to be taken into account: First, the players may make random decision errors from time to time. Second, the players' observations of each other's previous actions may be incorrect. The players will try to compensate for these errors based on the information they have. We examine convergence property of the game in such scenarios, and establish convergence to the equilibrium point under some mild assumptions when both players are restricted to two actions.

\end{abstract}

\section{Introduction}

Game theory has recently been used as an effective tool to model and solve many security problems in computer and communication networks. In a noncooperative matrix game between an Attacker and a Defender, if the payoff matrices are assumed to be known to both players, each player can compute the set of Nash equilibria of the game and play one of these strategies to maximize her expected gain (or minimize its expected loss)\footnote{The problem of each player choosing a Nash equilibrium out of multiple Nash equilibria is not discussed within the scope of this paper.}. However, in practice, the players do not necessarily have full knowledge of each other's payoff function. If the game is repeated, a mechanism called fictitious play (FP) can be used for each player to learn her opponent's motivations. In a FP process, each player observes all the actions and makes estimates of the mixed strategy of her opponent. At each stage, she updates this estimate and plays the pure strategy that is the best response (or generated based on the best response) to the current estimate of the other's mixed strategy. It can be seen that in a FP process, if one person plays a fixed strategy (either of the pure or mixed type), the other person's strategy will converge to the best response to this fixed strategy. Furthermore, it has been shown that, for many classes of games, such a FP process will finally render both players playing the Nash equilibrium.

In this paper, we examine a two-player game, where an Attacker (denoted as player $1$ or $P_1$) and a Defender (denoted as player $2$ or $P_2$) participate in a discrete-time repeated nonzero-sum matrix game. In a general setting, the Attacker has $m$ possible actions and the Defender has $n$ posssible actions to choose from. 
When such a security game is played between two automated systems over a network, in order to have a good model, we have to take into account several practical issues. First, the players may make random decision errors from time to time. Instead of playing an action $a_i^j$ that is the output of the best-response computation, player $i$ may play another action $a_i^k$ with some probability (which is typically small for functional systems). Second, the observation that each player makes on her opponent's actions may also be incorrect, which will definitely affect her own responding actions. There are many factors giving rise to these problems: The non-idealiality of electronic and software systems, the uncertain and noisy characteristic of observation data, and the erroneous nature of the channels on which commands and observations are communicated, to name a few. 

It is these scenarios that we aim to address in this paper. We examine convergence of players' strategies in the FP process with decision and observation errors. If these strategies do converge, we quantify the new Nash equilibrium and thus estimate how these decision and observation errors affect the learning process and the equilibrium of the game.

Security games have been examined extensively in a number of papers, see for example, \cite{ABC03,ABC04,ABI,ZY}. The work in \cite{LCMGN} employs the framework of Bayesian games to address the intrusion detection problem in wireless \textit{ad hoc} networks. In \cite{Chen}, the author examines the intrusion detection problem in heterogenous networks as a nonzero-sum static game. The work in \cite{Sall} addresses this problem using the framework of zero-sum stochastic games \cite{Owen}. In \cite{NABGN09}, we develop a network model based on \textit{linear influence networks} that allows us to take into consideration the correlation among the nodes in terms of both security assets and vulnerabilities. 

Relevant literature on fictitious play can be found in  \cite{ROB,MIS,UB,SA04,SA05,AO,MSA05}. For two-player zero-sum classical FP, the convergence proof was obtained for arbitrary numbers of actions for each player ($m \times n$) \cite{ROB}. For nonzero-sum games, the proofs for two-player FP have been found for the case where one player is restricted to two actions (See \cite{UB} for classical FP and \cite{SA04} for stochastic FP). In \cite{NABICC09}, we address the classical FP and stochastic FP with imperfect observations for the case where each player is restricted to two actions. 

Our contributions in this paper are as follows. First, we formulate the repeated security games where players make random decision errors as a fictitious play process. We discuss the convergence of such games in the general case with arbitrary numbers of actions for each player. We then establish the convergence property for several classes of games with decision errors where both players are restricted to two actions. Second, we examine the fictitious play process where the players' observations are imperfect and the players try to compensate for the observation errors. We again establish the convergence property for the case where both players are restricted to two actions. We point out a number of scenarios that can be considered as special cases of this result.

In Section \ref{BG}, we introduce some background and notation adopted from \cite{SA04}, \cite{SA05}. The analysis for the stochastic FP with decision errors is presented in Section \ref{sec:dec_errors}. In Section \ref{sec:obs_errors}, we address the FP with observation errors. Finally, some concluding remarks end the paper.

\section{Background} \label{BG}
\subsection{Static games} \label{static}
We present an overview of some concepts for static security games, where player $P_1$ has $m$ and player $P_2$ has $n$ possible actions. In equations written for the generic player $P_i, \ i=1,2$, we use $k$ to denote $m$ or $n$. Denote by $p_1 \in \Delta(m)$ and $p_2 \in \Delta(n)$ a pair of mixed strategies for $P_1$ and $P_2$, respectively, where $\Delta(k)$ is the simplex in $\Re^k$, i.e., 
  \begin{equation}
  \Delta(k) \equiv \left\{ s \in \Re^k \vert s_j \geq 0, j=1,\ldots,k, \ \sum_{j=1}^k s_j = 1 \right\}.
  \end{equation}
  The utility function of $P_i$, $U_i(p_i,p_{-i})$, is given by \footnote{As standard in the game theory literature, the index $-i$ is used to indicate those of other players, or the opponent's in this case.}
   \begin{eqnarray}
   \label{utility} U_i(p_i,p_{-i}) = p_i^T M_i p_{-i} + \tau_i H(p_i),
  \end{eqnarray}
  where $M_i$ is the payoff matrix of $P_i,i=1,2$, and $H: Int(\Delta(k)) \rightarrow \mathcal{R} $ is the entropy function of the probability vector $p_i$: $H(p_i)=-p_i^T log(p_i)$ (Note that $M_1$ is of dimension $m \times n$ and $M_2$ $n \times m$). The weighted entropy $\tau_i H(p_i)$ with $\tau_i \geq 0$ is introduced to boost mixed strategies. In a security game, $\tau_i$ represents how much player $i$ wants to randomize its actions, and thus is not necessarily known to the other player. Also, for $\tau_1=\tau_2=0$ (referred to as classical FP), the best response mapping can be set-valued, while it has a unique value when $\tau_i>0$ (referred to as stochastic FP) \cite{ZY} \cite{SA05}.
  For a static game, each player selects an integer action $a_i$ according to the mixed strategy $p_i$. The (instant) payoff for player $P_i$ is $v^T_{a_i} M_i v_{a_{-i}} + \tau_i H(p_i)$, where we use $v_j,j=1,\ldots,k$, to indicate the $j$th vertex of the simplex $\Delta(k)$ (For example, when $k=2$, $v_1=[1 \ 0]^T$ for the first action, and $v_2=[0 \ 1]^T$ for the second action). For a pair of mixed strategies $(p_1,p_2)$, the utility functions are given by the expected payoffs:
  \begin{equation}
  U_i(p_i,p_{-i})=E \left[ v^T_{a_i} M_i v_{a_{-i}} \right] + \tau_i H(p_i).
  \end{equation}
  Now,	the \textit{best response} mappings $\beta_1: \Delta(n) \rightarrow  \Delta(m)$ and $\beta_2: \Delta(m) \rightarrow  \Delta(n)$ are defined as:
  \begin{equation} \label{utility_max}
  \beta_i(p_{-i}) = \arg \max_{p_i \in \Delta(k)} {U_i(p_i,p_{-i})}.
  \end{equation}
  If $\tau_i>0$, from (\ref{utility_max}), the best response is unique as mentioned earlier, and is given by the soft-max function:
  \begin{equation} \label{best_response}
	\beta_i(p_{-i}) = \sigma \left( \frac{M_i p_{-i}}{\tau_i} \right),
  \end{equation}
  where the soft-max function $\sigma: \Re^k \rightarrow \ \textrm{Interior}(\Delta(k))$  is defined as
  \begin{equation} \label{softmax}
  (\sigma(x))_j = \frac{e^{x_j}} {\sum_{j=1}^k e^{x_j}}, j=1,\ldots,k.
  \end{equation}
  Note that $(\sigma(x))_j > 0$, and thus the range of the soft-max function is just the interior of the simplex. 
  
  Finally, a (mixed strategy) Nash equilibrium is defined to be a pair $(p_1^*, p_2^*) \in \Delta(m) \times \Delta(n)$ such that for all $p_1 \in \Delta(m)$ and $p_2 \in \Delta(n)$
  \begin{equation} \label{saddle_point}
  U_i(p_i,p^*_{-i}) \leq U_i(p^*_i,p^*_{-i}).
  \end{equation}
  We can also write a Nash equilibrium $(p^*_1, p^*_2)$ as the fixed point of the best response mappings:
  \begin{equation} \label{best_response_mapping}
  p^*_i = \beta _i (p^*_{-i}).
  \end{equation}  
\subsection{Fictitious play}  \label{ss:ctfp}
\subsubsection{Discrete-Time Fictitious Play} 
From the static game described in Subsection \ref{static}, we define discrete-time FP as follows. Suppose that the game is repeated at times $k \in \left\{0,1,2,\ldots \right\}$. The empirical frequency $q_i(k)$ of player $P_i$ is given by
\begin{equation} \label{em_freq}
q_i(k+1) = \frac{1}{k+1} \sum_{j=0}^k v_{a_i(j)} 
\end{equation}
Using induction, we can prove the following recursive relation:
\begin{equation} \label{rec_em_freq}
q_i(k+1)  =  \frac{k}{k+1} q_i(k) + \frac{1}{k+1} v_{a_i(k)}.
\end{equation}
At time $k$, player $P_i$ picks the best response to the empirical frequency of the opponent's actions:
  \begin{equation}
  p_i(k) = \beta_i (q_{-i}(k)).
  \end{equation}
\subsubsection{Continuous-Time Fictitious Play}  
From the equations of discrete-time FP (\ref{em_freq}), (\ref{rec_em_freq}), the continuous-time version of the iteration can be stated as follows (\cite{SA04}, \cite{SA05}, also see \cite{AO,NABICC09} for the derivation):
\begin{eqnarray} \label{ctfp}
 \dot{p}_i(t) &=& \beta_i(p_{-i}(t)) -  p_i(t), \ i=1,2.
\end{eqnarray}

\subsection{Algorithms} \label{algos}
We present in this subsection two algorithms for discrete-time stochastic FP. Algorithm \ref{sfp_po_algo}, derived from \cite{SA04,SA05,NABICC09}, is used for the case when players' observations are considered to be perfect or when they have no estimates of observation errors. Algorithm \ref{sfp_io_algo}, a generalized version of the one in \cite{NABICC09}, is used for players who have estimates of observation errors and want to compensate for these errors.
\subsubsection{Stochastic FP with perfect observations}  \label{sfp_po_algo}

In stochastic FP, at time $k$, player $i,i=1,2$, carries out the following steps:
\begin{enumerate}
	\item 
	Update the empirical frequency of the opponent using (\ref{rec_em_freq}).
	\item
	Compute the best response $\beta_i(q_{-i}(k))$ using (\ref{best_response}). (Note that the result is always a completely mixed strategy.)
	\item
	Generate an action $a_i(k)$ using the mixed strategy from step ($2$), $a_i(k)=rand \left[ \beta_i(q_{-i}(k)) \right]$, where we use $rand$ to denote the randomizer function that gives $a_i(k)$ such that the expectation $E \left[a_i(k) \right]=\beta_i(q_{-i}(k))$.
\end{enumerate}
\subsubsection{Stochastic FP with imperfect observations} \label{sfp_io_algo}

At time $k$, player $i,i=1,2$, carries out the following steps:
\begin{enumerate}
	\item 
	Update the observed frequency of the opponent $\overline{q}_{-i}$ using (\ref{rec_em_freq}).
	\item
	Compute the estimated frequency
	\begin{equation}
	{q}_{-i}=f_{-i}(\overline{q}_{-i}).
	\end{equation}
	\item
	Compute the best response $\beta_i(q_{-i}(k))$ using (\ref{best_response}). (Note that the result is always a completely mixed strategy.)
	\item
	Generate an action $a_i(k)$ using the mixed strategy from step ($3$), $a_i(k)=rand[\beta_i(q_{-i}(k))]$.
\end{enumerate}
\subsection{A convergence result for $m=n=2$ with perfect observations}
We restate the following theorem from \cite{SA04,NABICC09}, for the general case where the coefficients of the entropy terms for the players ($\tau_1$ and $\tau_2$) are not necessarily equal (Cf. Equation (\ref{utility})). This theorem in \cite{SA04} is stated for $\tau_1 = \tau_2$, however, one can always scale the payoff matrices to get the general case.
\begin{theorem}  \label{t2sa04}
(A variant of Theorem $3.2$ \cite{SA04} for general $\tau_1,\tau_2>0$.)
Consider a two-player two-action fictitious play process with $(L^T \tilde{M}_1 L)(L^T \tilde{M}_2 L) \neq 0$, where $\tilde{M}_i$ are the payoff matrices of $P_i,\ i=1,2$ and $L:= (1, \ -1)^T$. The solutions of continuous-time FP (\ref{ctfp}) satisfy
\begin{eqnarray}
 \lim_{t \rightarrow \infty} \left( p_1(t) - \beta_1(p_2(t))   \right) &=& 0 \\
\lim_{t \rightarrow \infty} \left( p_2(t) - \beta_2(p_1(t))   \right) &=& 0,
\end{eqnarray}
where $\beta_i(p_{-i}),\ i=1,2$, are given in (\ref{best_response}).
\end{theorem}

\section{Security games with decision errors} \label{sec:dec_errors}
In this section, we consider the situations where players are not totally rational or the channels carrying commands are error prone.  Specifically, $P_1$ makes decision errors with probabilities	$\alpha_{ij}$'s where $\alpha_{ij},\ i,j=1 \ldots m$, is the probability that $P_1$ intends to play action $i$ but ends up playing action $j$, $\alpha_{ij} \geq 0$, $\sum_{j=1}^m \alpha_{ij}=1,\ i=1 \ldots m$. Similarly, $P_2$'s decision error probabilities are given by $\epsilon_{ij},\ \epsilon_{ij} \geq 0$, $\sum_{j=1}^m \epsilon_{ij}=1,\ i=1 \ldots n$. This is called \textit{``trembling hand''} problem in the game theory literature (See for example, Reference \cite{TB}, Subsection 3.5.5). The decision error matrices $D_1$ and $D_2$ are given below.
		\begin{eqnarray} \label{error_dec_mn_1}
		D_1 &=& \left(
		\begin{array}{cccc}
			\alpha_{11}	& \alpha_{12} & \ldots & \alpha_{1m}  \\
			\alpha_{21}	& \alpha_{22} & \ldots & \alpha_{2m}  \\	
			\ldots \\
			\alpha_{m1}	& \alpha_{m2} & \ldots & \alpha_{mm}  \\	
		\end{array}
		\right), \\
\label{error_dec_mn_2} 	 
		D_2 &=& \left(
		\begin{array}{cccc}	
			\epsilon_{11}	& \epsilon_{12} & \ldots & \epsilon_{1n}  \\
			\epsilon_{21}	& \epsilon_{22} & \ldots & \epsilon_{2n}  \\	
			\ldots \\
			\epsilon_{n1}	& \epsilon_{n2} & \ldots & \epsilon_{nn}  \\	
				\end{array}
		\right).
		\end{eqnarray}	

 When $m=n=2$, the decision error matrices can be written as:
		\begin{equation} \label{error_dec}
		D_1 = \left(
		\begin{array}{cc}
			1-\alpha	&	\gamma \\
			\alpha &	1-\gamma \\
		\end{array}
		\right), \  \
		D_2 = \left(
		\begin{array}{cc}
			1-\epsilon	&	\mu \\
			\epsilon & 1-\mu	 \\
		\end{array}
		\right)
		\end{equation}	
The decision errors of each player in this case are illustrated in Figure \ref{fig:decision_errors}. In what follows, we  state two standard results in digital communications. The proofs are similar to those for the case $m=n=2$ in \cite{NABICC09}.
\begin{proposition} \label{dec_error_limit}
Consider the two-player discrete-time fictitious play with decision errors where the error probabilities are given in Equations (\ref{error_dec_mn_1}) and (\ref{error_dec_mn_2}). Let $\widetilde{\alpha}_{ij}, \ i,j=1 \ldots m$, and $\widetilde{\epsilon}_{ij}, \ i,j=1 \ldots n$, be the empirical decision error frequencies of $P_1$ and $P_2$, respectively. If decision errors are assumed to be independent from stage to stage, it holds that
\begin{eqnarray}
\nonumber \lim_{k \rightarrow \infty} \ a.s. \ \widetilde{\alpha}_{ij}  &=& \alpha_{ij}, \ i,j=1 \ldots m,\\
\lim_{k \rightarrow \infty} \ a.s. \ \widetilde{\epsilon}_{ij} &=& {\epsilon}_{ij}, \ i,j=1 \ldots n.
\end{eqnarray}
where we use \emph{$\lim$ a.s.} to denote \emph{almost sure convergence}.
\end{proposition}
\begin{proposition} \label{dec_error_channel}
Consider a two-player discrete-time fictitious play with decision errors where the error probabilities are given in Equations (\ref{error_dec_mn_1}) and (\ref{error_dec_mn_2}). Let $\overline{q}_{i}$ be the empirical frequency of player $i$'s real actions and $q_{i}$ be the frequency of player $i$'s intended actions (generated from the best response at each stage). If decision errors are assumed to be independent from stage to stage, it holds that
\begin{eqnarray} \label{dec_error_freq_prob}
\lim_{k \rightarrow \infty} \ a.s. \ \overline{q}_{i} = D_{i} (\lim_{k \rightarrow \infty} \ a.s. \ q_{i}), \ i=1,2,
\end{eqnarray}
where $D_i$ are the decision error matrices given in Equations (\ref{error_dec_mn_1}) and (\ref{error_dec_mn_2}).
		\end{proposition}
 
\begin{psfrags}
\psfrag{alpha}{$\alpha$}
\psfrag{1-alpha}{$1-\alpha$}
\psfrag{gamma}{$\gamma$}
\psfrag{1-gamma}{$1-\gamma$}
\psfrag{epsilon}{$\epsilon$}
\psfrag{1-epsilon}{$1-\epsilon$}
\psfrag{mu}{$\mu$}
\psfrag{1-mu}{$1-\mu$}
\psfrag{a11}{$a_1^1$ (A)}
\psfrag{a12}{$a_1^2$ (N)}
\psfrag{a21}{$a_2^1$ (D)}
\psfrag{a22}{$a_2^2$ (N)}
\psfrag{Attacker}{$P_1$ (Attacker)}
\psfrag{Defender}{$P_2$ (Defender)}
\begin{figure}[ht]
  \centering
  \includegraphics[width=6.5cm]{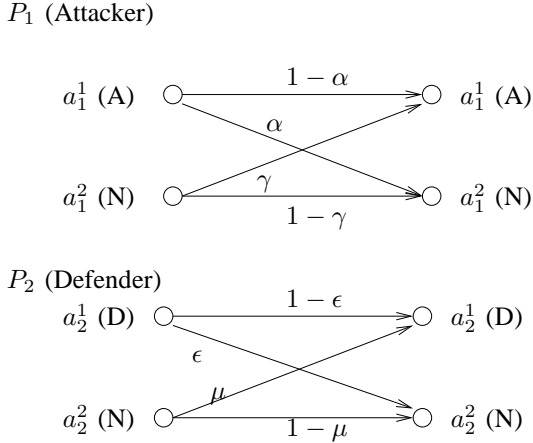}
  \caption{The case $m=n=2$ where players make decision errors with probabilities  $\alpha, \gamma, \epsilon$, and $\mu$.} \label{fig:decision_errors}
\end{figure}
\end{psfrags}
\subsection{If the players know their own decision error probabilities}
We first consider the case where the players both have complete information about the decision error matrices $D_i$, $i=1,2$. If they both also know the payoff matrices $M_i, \ i=1,2$, then each can compute and play one of the Nash equilibria right from beginning. The problem then can be considered as a stochastic version of the trembling hand problem. Specifically, suppose that each player still wants to randomize their empirical frequency $\overline{p}_{i}$ (instead of the frequency of their intended actions, or intended frequency, ${p}_{i}$) by including an entropy term in their utility function, we have that
   \begin{eqnarray}
   \label{utility_th} U_i(p_i,p_{-i}) = p_i^T \tilde{M}_i p_{-i} + \tau_i H(D_i p_i), \ i=1,2,
  \end{eqnarray}
  where $p_i$'s are intended frequencies, $\tilde{M}_1=D^T_1 M_1 D_2$ and $\tilde{M}_2=D^T_2 M_2 D_1$ (These are the payoff matrices resulted from decision errors using the results in Propositions \ref{dec_error_limit} and \ref{dec_error_channel}, see for example \cite{TB} for derivation). Using $\overline{p}_{i} := D_i {p}_{i},\ i=1,2$, the utility functions now can be written as
   \begin{eqnarray}
   U_i(p_i,p_{-i}) = \overline{p}_{i}^T M_i \overline{p}_{-i} + \tau_i H(\overline{p}_i), \ i=1,2.
  \end{eqnarray}
  The game is thus reduced to the one without decision errors and the Nash Equilibrium of the static game is known from Subsection \ref{static} to satisfy:
  \begin{equation} \label{best_response_mapping_decision_error}
  \overline{p}^*_i = \beta_i (\overline{p}^*_{-i}), \ i=1,2,
  \end{equation}  
  or equivalently (with the assumption that $D_i$'s are invertible): 
  \begin{equation}
  {p}^*_i = (D_i)^{-1} \beta_i (D_{-i} {p}^*_{-i}), \ i=1,2.
  \end{equation}  
	The best response is now given as
  \begin{equation} \label{best_response_decision_error}
	{p}_i = (D_i)^{-1} \beta_i(\overline{p}_{-i}) = (D_i)^{-1} \sigma \left( \frac{M_i \overline{p}_{-i}}{\tau_i} \right).
  \end{equation}
In the corresponding FP process (the \textit{``trembling hand stochastic FP''}), as each player $P_i$ can observe her opponent's empirical frequency $\overline{p}_{-i}$, she does not need to know $D_{-i}$ to compute the best response. We thus state below a convergence result for the FP process with decision errors for the case $m=n=2$.
\begin{proposition}
Consider a two-player two-action fictitious play process where players make decision errors with invertible decision error matrices $D_1$ and $D_2$, respectively. Suppose that at each step, each player calculates the best response taking into account their own decision errors using Equation (\ref{best_response_decision_error}). If  $(L^T M_1 L)(L^T M_2 L) \neq 0$, $L:= (1, \ -1)^T$, the solutions of the  continuous-time FP process with decision errors will satisfy
\begin{eqnarray}
\nonumber \lim_{t \rightarrow \infty} p_1(t) = D_1^{-1} \sigma \left( \frac{ M_1 D_2 \lim_{t \rightarrow \infty} p_{2}(t)}{\tau_1} \right) , \\ 
\label{dec_err} \lim_{t \rightarrow \infty} p_2(t) =D_2^{-1} \sigma \left( \frac{M_2 D_1 \lim_{t \rightarrow \infty} p_{1}(t)}{\tau_2} \right).
\end{eqnarray}
where $\sigma(.)$ is the soft-max function defined in (\ref{softmax}).
\end{proposition}
\begin{proof}
The proof can be obtained using Theorem \ref{t2sa04} and the fact $\overline{p}_{i} := D_i {p}_{i},\ i=1,2$.
\end{proof}

It thus can be seen that with knowledge of their own decision errors, players can completely precompensate for these errors and the equilibrium empirical frequencies remain the same as those of the original game without decision errors.
\subsection{If the players are unaware of all the decision error probabilities} \label{sub:dec_errors_unaware}
We consider in this subsection a two-player fictitious play process with decision errors where the decision error probabilities are not known to both players. Each player plays the regular stochastic FP Algorithm \ref{sfp_po_algo}. We are interested in whether or not the FP process will converge, and when it does, what the equilibrium will be. We first examine the general case with arbitrary $m,\ n$, and then the special case where $m=n=2$. We first use Proposition \ref{dec_error_channel} and the same arguments as in the proof of Theorem 3 \cite{NABICC09} to approximate the discrete-time FP with the continuous-time version. At time step $k$, as each player $P_i$ generates her action $v_{a_i(k)}$ based on the best response to her opponent's empirical frequency $\overline{q}_{-i}$,  the expectation of $v_{a_i(k)},\ i=1,2$, will be given by
\begin{eqnarray} \label{}
\nonumber E \left[v_{a_1(k)} \right]  &=&  D_1 \beta_1(\overline{q}_2(k)),\\
\nonumber E[v_{a_2(k)}]  &=&  D_2 \beta_2(\overline{q}_1(k)),
\end{eqnarray}
where $D_1$ and $D_2$ account for decision errors.
The mean dynamic of the empirical frequencies then can be written as follows
\begin{eqnarray} \label{}
\nonumber \overline{q}_1(k+1)  &=&  \frac{k}{k+1} \overline{q}_1(k) + \frac{1}{k+1} D_1 \beta_1(\overline{q}_2(k)),\\
\label{io_em_freq} \overline{q}_2(k+1)  &=&  \frac{k}{k+1} \overline{q}_2(k) + \frac{1}{k+1}  D_2 \beta_2(\overline{q}_1(k)).
\end{eqnarray}
From the mean dynamic, we can derive the continuous-time approximation (See \cite{NABACC10TR} for the derivation):
\begin{eqnarray} 
\nonumber  \dot{\overline{p}}_1(t) &=& D_1 \beta_1(\overline{p}_{2}(t)) -  \overline{p}_1(t),\\
\label{ctde}  \dot{\overline{p}}_2(t) &=& D_2 \beta_2(\overline{p}_{1}(t)) -  \overline{p}_2(t).
\end{eqnarray}
It can be seen that a pair of mixed strategies $(p^*_1, p^*_2)$ that satisfies
\begin{eqnarray} \label{ct_mn_dynamics}
\nonumber  \overline{p}^*_1(t) & =& D_1 \beta_1( \overline{p}^*_{2}(t)),\\
\nonumber  \overline{p}^*_2(t) &=& D_2 \beta_2( \overline{p}^*_{1}(t)).
\end{eqnarray}
will be an equilibrium point of the dynamics (\ref{ctde}). For some results on the stability of the equilibrium point in the continuous-time system and the discrete-time system for general values of $m$ and $n$, we refer to \cite{NABACC10TR}. When $m=n=2$, it turns out the point $(\overline{p}^*_1, \overline{p}^*_2)$ is globally stable for the continuous-time system under some mild assumptions. We thus state the following theorem for this special case.
\begin{theorem} \label{io_ctfp_effect}
Consider a two-player two-action fictitious play process where players make decision errors with decision error matrices $D_1$ and $D_2$, respectively. Suppose that the players are unaware of all the decision error probabilities and use the regular stochastic FP algorithm \ref{sfp_po_algo}. If $D_i, \ i=1,2$, are invertible and $(L^TM_1D_2 L)(L^TM_2D_1 L) \neq 0$, the solutions of continuous-time FP process with decision errors (\ref{ctde}) will satisfy
\begin{eqnarray}
\nonumber \lim_{t \rightarrow \infty} \overline{p}_1(t) =D_1 \sigma \left( \frac{M_1 \lim_{t \rightarrow \infty} \overline{p}_{2}(t)}{\tau_1} \right), \\ 
\label{dec_err} \lim_{t \rightarrow \infty} \overline{p}_2(t) =D_2 \sigma \left( \frac{M_2 \lim_{t \rightarrow \infty} \overline{p}_{1}(t)}{\tau_2} \right).
\end{eqnarray}
where $\sigma(.)$ is the soft-max function defined in (\ref{softmax}).
\end{theorem}
\begin{proof}
The proof, some remarks, and a numerical example can be found in \cite{NABACC10TR}.
\end{proof}
\section{Security games with observation errors} \label{sec:obs_errors}

In \cite{NABICC09}, we study the effect of observation errors on convergence to the NE in a $2 \times 2$ FP process. We also prove that if each player has a correct estimate of error probabilities of observations, they can reverse the effect of the channel to obtain the NE of the original static game. In this section, we present a generalized version of these results.
Consider a two-player fictitious play game with imperfect observations where the error channels are given in Equations (\ref{error_channel_mn_1}) and (\ref{error_channel_mn_2}). 
		\begin{eqnarray} \label{error_channel_mn_1}
		C_1 &=& \left(
		\begin{array}{cccc}
			\alpha_{11}	& \alpha_{12} & \ldots & \alpha_{1m}  \\
			\alpha_{21}	& \alpha_{22} & \ldots & \alpha_{2m}  \\	
			\ldots \\
			\alpha_{m1}	& \alpha_{m2} & \ldots & \alpha_{mm}  \\	
		\end{array}
		\right), \\
\label{error_channel_mn_2} 	 
		C_2 &=& \left(
		\begin{array}{cccc}	
			\epsilon_{11}	& \epsilon_{12} & \ldots & \epsilon_{1n}  \\
			\epsilon_{21}	& \epsilon_{22} & \ldots & \epsilon_{2n}  \\	
			\ldots \\
			\epsilon_{n1}	& \epsilon_{n2} & \ldots & \epsilon_{nn}  \\	
				\end{array}
		\right),
		\end{eqnarray}	
		where $\alpha_{ij},\ i,j=1 \ldots m$ is the probability that $P_1$'s action $i$ is erroneously observed as action $j$, $\alpha_{ij} \geq 0$, $\sum_{j=1}^m \alpha_{ij}=1,\ i=1 \ldots m$, and $\epsilon_{ij},\ i,j=1 \ldots n$ is the probability that $P_2$'s action $i$ is erroneously observed as action $j$, $\epsilon_{ij} \geq 0$, $\sum_{j=1}^m \epsilon_{ij}=1,\ i=1 \ldots n$.
Suppose that the players have their estimates of the errror probabilities as follows:
		\begin{eqnarray} \label{error_channel_mn_est_1}
		\overline{C}_1 &=& \left(
		\begin{array}{cccc}
			\overline{\alpha}_{11}	& \overline{\alpha}_{12} & \ldots & \overline{\alpha}_{1m}  \\
			\overline{\alpha}_{21}	& \overline{\alpha}_{22} & \ldots & \overline{\alpha}_{2m}  \\	
			\ldots \\
			\overline{\alpha}_{m1}	& \overline{\alpha}_{m2} & \ldots & \overline{\alpha}_{mm}  \\	
		\end{array}
		\right), \\
		 \label{error_channel_mn_est_2}
		\overline{C}_2 &=& \left(
		\begin{array}{cccc}	
			\overline{\epsilon}_{11}	& \overline{\epsilon}_{12} & \ldots & \overline{\epsilon}_{1n}  \\
			\overline{\epsilon}_{21}	& \overline{\epsilon}_{22} & \ldots & \overline{\epsilon}_{2n}  \\	
			\ldots \\
			\overline{\epsilon}_{n1}	& \overline{\epsilon}_{n2} & \ldots & \overline{\epsilon}_{nn}  \\	
				\end{array}
		\right),
		\end{eqnarray}	
		where $\overline{\alpha}_{ij} \geq 0$, $\sum_{j=1}^m \overline{\alpha}_{ij}=1,\ i=1 \ldots m$, and $\overline{\epsilon}_{ij} \geq 0$, $\sum_{j=1}^m \overline{\epsilon}_{ij}=1,\ i=1 \ldots n$.
\begin{psfrags}
\psfrag{alpha}{$\alpha$}
\psfrag{1-alpha}{$1-\alpha$}
\psfrag{gamma}{$\gamma$}
\psfrag{1-gamma}{$1-\gamma$}
\psfrag{epsilon}{$\epsilon$}
\psfrag{1-epsilon}{$1-\epsilon$}
\psfrag{mu}{$\mu$}
\psfrag{1-mu}{$1-\mu$}
\psfrag{a11}{$a_1^1$ (A)}
\psfrag{a12}{$a_1^2$ (N)}
\psfrag{a21}{$a_2^1$ (D)}
\psfrag{a22}{$a_2^2$ (N)}
\psfrag{Attacker}{$P_1$ (Attacker)}
\psfrag{System}{$P_2$ (Defender)}
\psfrag{Actions}{Actions}
\psfrag{Observations}{Observations}
\begin{figure}[ht]
  \centering
  \includegraphics[width=6.5cm]{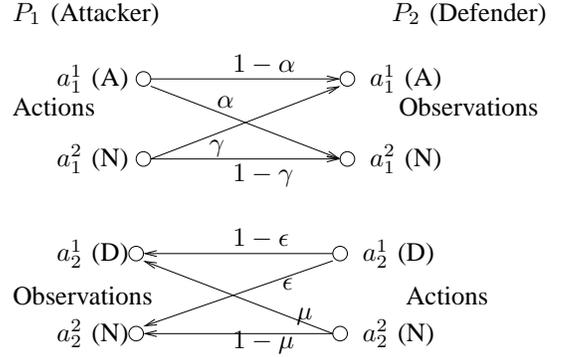}
  \caption{Players observe their opponent's actions through binary channels with error probabilities  $\alpha, \gamma, \epsilon$, and $\mu$.} \label{fig:ascs}
\end{figure}
\end{psfrags}
We first restate Propositions \ref{dec_error_limit} and \ref{dec_error_channel} in the context of repeated games with imperfect observations.
\begin{proposition} \label{obs_error_limit}
Consider the two-player discrete-time fictitious play with imperfect observations where error probabilities are given in Equations (\ref{error_channel_mn_1}) and (\ref{error_channel_mn_2}). Let $\widetilde{\alpha}_{ij}, \ i,j=1 \ldots m$, and $\widetilde{\epsilon}_{ij}, \ i,j=1 \ldots n$, be the empirical error frequencies of observations on $P_1$'s and $P_2$'s actions, respectively. If channel errors are assumed to be independent from stage to stage, it holds that
\begin{eqnarray}
\nonumber \lim_{k \rightarrow \infty} \ a.s. \ \widetilde{\alpha}_{ij}  &=& \alpha_{ij}, \ i,j=1 \ldots m,\\
\lim_{k \rightarrow \infty} \ a.s. \ \widetilde{\epsilon}_{ij} &=& {\epsilon}_{ij}, \ i,j=1 \ldots n.
\end{eqnarray}
where we use \emph{$\lim$ a.s.} to denote \emph{almost sure convergence}.
\end{proposition}
\begin{proposition} \label{obs_error_channel}
Consider the two-player discrete-time fictitious play with imperfect observations where error probabilities are given in Equations (\ref{error_channel_mn_1}) and (\ref{error_channel_mn_2}). Let $\overline{q}_{i}$ be the observed frequency and $q_{i}$ be the empirical frequency of player $i$. If channel errors are assumed to be independent from stage to stage, it holds that
\begin{eqnarray} \label{error_freq_prob}
\lim_{k \rightarrow \infty} \ a.s. \ \overline{q}_{i} = C_{i} (\lim_{k \rightarrow \infty} \ a.s. \ q_{i}), \ i=1,2,
\end{eqnarray}
where $C_i$ are the channel error matrices given in Equations (\ref{error_channel_mn_1}) and (\ref{error_channel_mn_2}).
		\end{proposition}

If both players have their estimates of the errror probabilities as in Equations (\ref{error_channel_mn_est_1}) and (\ref{error_channel_mn_est_2}), they can play the stochastic FP algorithm given in \ref{sfp_io_algo} with $f_{-i}(\overline{q}_{-i})=(\overline{C}_i)^{-1}\overline{q}_{-i}$ to compensate for observation errors (Using the results in Propositions \ref{obs_error_limit} and \ref{obs_error_channel}). 
Again we can use the same procedure as in Subsection \ref{sub:dec_errors_unaware} to approximate the discrete-time FP with the continuous-time version. 
\begin{eqnarray} \label{}
\nonumber q_1(k+1)  &=&  \frac{k}{k+1} q_1(k) \\
\nonumber && + \ \frac{1}{k+1} \sigma \left( \frac{M_1 (\overline{C}_2)^{-1} C_2  q_{2}(k)}{\tau_1} \right),\\
\nonumber q_2(k+1)  &=&  \frac{k}{k+1} q_2(k) \\
\nonumber &&+ \ \frac{1}{k+1} \sigma \left( \frac{M_2 (\overline{C}_1)^{-1} C_1  q_{1}(k)}{\tau_2} \right).
\end{eqnarray}
The continuous-time approximation is given by:
\begin{eqnarray} 
\nonumber \dot{p}_1(t) & =& \sigma \left( \frac{M_1 (\overline{C}_2)^{-1} C_2 p_{2}(t)}{\tau_1} \right) -  p_1(t),\\
\label{ct_mn_dynamics} \dot{p}_2(t)  &=& \sigma \left( \frac{M_2 (\overline{C}_1)^{-1} C_1 p_{1}(t)}{\tau_2} \right) -  p_2(t).
\end{eqnarray}
It can be seen that a pair of mixed strategies $(q^*_1, q^*_2)$ that satisfies
\begin{eqnarray}
\nonumber  p^*_1(t) & =& \sigma \left( \frac{M_1 (\overline{C}_2)^{-1} C_2 p^*_{2}(t)}{\tau_1} \right),\\
\nonumber p^*_2(t) &=& \sigma \left( \frac{M_2 (\overline{C}_1)^{-1} C_1 p^*_{1}(t)}{\tau_2} \right).
\end{eqnarray}
will be an equilibrium point of the dynamics (\ref{ct_mn_dynamics}). 
For some results on the stability of the equilibrium point in the continuous-time system and the discrete-time system for general values of $m$ and $n$, we refer to \cite{NABACC10TR}. When $m=n=2$, again the point $(p^*_1, p^*_2)$ is globally stable for the continuous-time system under some mild assumptions. We have the following theorem.
\begin{theorem} \label{io_ctfp_inacc}
Consider a two-player two-action fictitious play game with imperfect observations where the error channels are given in Figure \ref{fig:ascs} and Equation (\ref{error_channel_22}). 
		\begin{equation} \label{error_channel_22}
		C_1 = \left(
		\begin{array}{cc}
			1-\alpha	&	\gamma \\
			\alpha &	1-\gamma \\
		\end{array}
		\right), \  \
		C_2 = \left(
		\begin{array}{cc}
			1-\epsilon	&	\mu \\
			\epsilon & 1-\mu	 \\
		\end{array}
		\right)
		\end{equation}	
Suppose that the players have their estimates of the errror probabilities as follows:
 	\begin{equation} \label{error_channel_est}
		\overline{C}_1 = \left(
		\begin{array}{cc}
			1-\overline{\alpha}	&	\overline{\gamma} \\
			\overline{\alpha} &	1-\overline{\gamma} \\
		\end{array}
		\right), \  \
		\overline{C}_2 = \left(
		\begin{array}{cc}
			1- \overline{\epsilon}	&	\overline{\mu} \\
			\overline{\epsilon} & 1-\overline{\mu}	 \\
		\end{array}
		\right)
		\end{equation}	
The players then play the stochastic FP given in \ref{sfp_io_algo}. If $(L^TM_1 (\overline{C}_2)^{-1} C_2 L)(L^T M_2 (\overline{C}_1)^{-1} C_1 L) \neq 0$, the solutions of continuous-time FP with imperfect observations (\ref{ctfp}) will satisfy
\begin{eqnarray} 
\nonumber \lim_{t \rightarrow \infty} p_1(t) = \sigma \left( \frac{M_1 (\overline{C}_2)^{-1} C_2 \lim_{t \rightarrow \infty} p_{2}(t)}{\tau_1} \right) , \\ 
\label{error_channel_est_theorem} \lim_{t \rightarrow \infty} p_2(t) = \sigma \left( \frac{M_2 (\overline{C}_1)^{-1} C_1 \lim_{t \rightarrow \infty} p_{1}(t)}{\tau_2} \right).
\end{eqnarray}
\end{theorem}
where $\sigma(.)$ is the soft-max function defined in (\ref{softmax}).
\begin{proof}
The proof, some remarks, and a numerical example can be found in \cite{NABACC10TR}.
\end{proof}
\section{Conclusion}

In this paper, we have introduced and discussed some repeated security game models that take into account players' decision errors and observation errors. Each player does not have access to her opponent's payoff matrix and thus has to learn this through the fictitious play process. However, in a practical setting, each player is expected to make random decision errors from time to time and also has to respond to imperfectly observed actions of the other player. We have studied the convergence property of such games and, if the FP process does converge, quantified the new equilibrium. Such analyses will help provide guidelines for players to maximize their gain or minimize their loss in a nonideal environment. 

We normally start from the mean dynamics of the discrete-time version of a game, proceed to continuous-time approximation and then analyze convergence of this continuous-time version. Although the convergence of the continuous-time fictitious play does not guarantee the almost sure convergence of the discrete-time counterpart, it does provide the necessary limiting results for the discrete-time version. 
\section{ACKNOWLEDGMENTS}

We would like to thank Deutsche Telekom Laboratories and the Boeing Company for their support. We are also grateful to the anonymous reviewers for their valuable comments.



\begin{thebibliography}{99}

\bibitem{ABC03}
T.~Alpcan and T.~Ba{\c{s}}ar, ``A Game Theoretic Approach to Decision and Analysis in Network Intrusion Detection'', \textit{Proceedings of the 42nd IEEE Conference on Decision and Control}, Hawaii, USA, 2003, pp. 2595--2600.

\bibitem{ABC04}

T.~Alpcan and T.~Ba{\c{s}}ar, ``A game theoretic analysis of intrusion detection in access control systems,'' \textit{Proceedings of the 43rd IEEE Conference on Decision and Control}, Paradise Island, Bahamas, 2004, pp. 1568--1573.

\bibitem{ABI}

T.~Alpcan and T.~Ba{\c{s}}ar, ``An intrusion detection game with limited observations,'' \textit{Proceedings of the 12th Int. Symp. on Dynamic Games and Applications}, Sophia Antipolis, France, 2006.

\bibitem{ZY}
Z.~Yin, ``\textit{Trust-based game-theoretic intrusion detection},'' M.S. thesis, University of Illinois at Urbana-Champaign, 2006. 

\bibitem{LCMGN}

Y.~Liu, C.~Comaniciu, and H.~Man, ``A Bayesian game approach for intrusion detection in wireless ad hoc networks,'' \textit{Proceedings of the Workshop on Game Theory for Networks (GameNets)}, Pisa, Italy, 2006.  

\bibitem{Chen}

L.~Chen, ``\textit{On Selfish and Malicious Behaviors in Wireless Networks - A Non-cooperative Game Theoretic Approach},'' Ph.D. thesis, Telecom ParisTech, 2008.

\bibitem{Sall}

K.~Sallhammar, ``\textit{Stochastic Models for Combined Security and Dependability Evaluation},'' Ph.D. thesis, Norwegian University of Science and Technology, 2007.

\bibitem{Owen}

G.~Owen,
\newblock {\textit {Game Theory}},
\newblock 3nd Ed., California: Academic Press, 2001.

\bibitem{NABGN09}

K.~C.~Nguyen, T.~Alpcan, and T.~Ba\c{s}ar, ``Stochastic Games for Security in Networks with Interdependent Nodes,'' \textit{Proc. of  the International Conference on Game Theory for Networks (GameNets 2009)}, Istanbul, Turkey, May 13-15, 2009.

\bibitem{ROB}

J.~Robinson, ``An iterative method of solving a game,'' \textit{Ann. Math.}, vol. 54, pp. 296–-301, 1951.

\bibitem{MIS}

K.~Miyasawa, ``On the convergence of learning processes in a $2 \times 2$ nonzero-sum two person game,''\textit{ Econometrics Research Program}, Princeton University, Research Momorandum No. 33, October 1961. Available at
http://www.princeton.edu/~erp/ERParchives/archivepdfs/M33.pdf

\bibitem{UB}

U.~Berger, ``Fictitious play in $2\times n$ games,'' \textit{Economics Working Paper Archive at WUSTL}. Available at:
http://econpapers.repec.org/paper/wpawuwpga/0303009.htm

\bibitem{SA04}
J.~S.~Shamma and G.~Arslan, ``Unified convergence proofs of continuous-time fictitious play,'' \textit{IEEE Transactions on Automatic Control}, Vol. 49, No.7, July 2004.

\bibitem{SA05}
J.~S.~Shamma and G.~Arslan, ``Dynamic fictitious play, dynamic gradient play, and distributed convergence to Nash equilibria,'' \textit{IEEE Transactions on Automatic Control}, Vol. 50, No.3, March 2005.

\bibitem{AO}

A.~Ozdaglar, ``Lecture 11: Fictitious Play and Extensions,'' \textit{Game Theory and Mechanism Design, MIT OpenCourseWare}. Available at
http://ocw.mit.edu/OcwWeb/Electrical-Engineering-and-Computer-Science/6-972Spring-2005/CourseHome/index.htm.

\bibitem{MSA05}

S.~Mannor, J.~S.~Shamma, G.~Arslan, ``Online calibrated forecasts: Efficiency vs universality for learning in games,'' \textit{Machine Learning}, Special Issue on Learning and Computational Game Theory, Vol. 67, No. 1--2, pp. 77--115,
May 2007.

\bibitem{TB}

T.~Ba\c{s}ar and G.~J. Olsder, \textit{Dynamic Noncooperative Game Theory}, 2nd ed. Philadelphia, PA: SIAM, 1999.

\bibitem{KHA}
H.~K.~Khalil, \textit{Nonlinear systems}, 2nd ed. Upper Saddle River, NJ: Prentice-Hall, 1998. 

\bibitem{NABICC09}

K.~C.~Nguyen, T.~Alpcan, and T.~Ba\c{s}ar, ``Security Games with Incomplete Information,'' \textit{Proc. of the 2009 IEEE International Conference on Communications (ICC 2009)}, Dresden, Germany, June 14--18, 2009.

\bibitem{NABACC10TR}

K.~C.~Nguyen, T.~Alpcan, and T.~Ba\c{s}ar, ``Security Games with Decision and Observation Errors,''  \textit{Technical Report}, UIUC, Oct. 2009. Available at http://decision.csl.uiuc.edu/\~{ }knguyen4/research/research.html.

\end{thebibliography}
\end{document}